\documentclass[a4paper,12pt]{article}
\usepackage{amsmath}
\usepackage{amsthm}
\usepackage{amsfonts}
\usepackage{mathrsfs}
\usepackage{graphicx}
\usepackage{hyperref}
\usepackage{float}

\newcommand{\lR}{\mathrm{I\hspace{-0.7mm}R}}

\newtheorem{theorem}{Theorem}
\newtheorem{lemma}{Lemma}

\newtheorem{assumption}{Assumption}

\numberwithin{equation}{section}

\hypersetup{ pdftitle={Report 2},
pdfauthor={Fiki T. Akbar}, 
pdfkeywords={Scalar Wave Equation, Non minimal Coupling, Existence}, 
bookmarksnumbered, pdfstartview={FitH}, urlcolor=blue, }

\begin{document}
\pagestyle{plain}




\title{\LARGE\textbf{Local and Global Existence of Solutions to Scalar Equations on Spatially Flat Universe as a Background with Non-minimal Coupling}}

\author{{Fiki T. Akbar$^{1,2}$, Bobby E. Gunara$^{1,2}$\footnote{Corresponding author}, Muhammad Iqbal$^{2}$, Hadi Susanto$^3$}\\ \\
$^{1}$\textit{\small Indonesian Center for Theoretical and
Mathematical Physics (ICTMP)}\\
$^{2}$\textit{\small Theoretical Physics Laboratory}\\
\textit{\small Theoretical High Energy Physics and Instrumentation Research Group,}\\
\textit{\small Faculty of Mathematics and Natural Sciences,}\\
\textit{\small Institut Teknologi Bandung}\\
\textit{\small Jl. Ganesha no. 10 Bandung, Indonesia, 40132}\\
$^{3}$\textit{\small Department of Mathematical Sciences, University of Essex,}\\
\textit{\small Colchester, CO4 3SQ, United Kingdom}\\ \\
\small email: ftakbar@fi.itb.ac.id, bobby@fi.itb.ac.id,\\ \small muhammad.iqbal7@students.itb.ac.id, hsusanto@essex.ac.uk}

\date{}

\maketitle




\begin{abstract}

We prove the wellposedness of scalar wave equations on spatially flat universe as a background with nonminimal coupling with the scalar potential turned on by introducing the $k$-order linear energy and the corresponding energy norm. In the local case, we show that both the $k$-order linear energy and the  energy norm are bounded for finite time with initial data in $H^{k+1}\times H^{k}$. Whereas in the global case, we have to add three assumptions  related to the nonminimal coupling constant, the scale factor of spacetimes, and the form of the scalar potential that has to be a polynomial with a small positive parameter. Then, we show that the solution does globally exist with a particular decay estimate that depends on the scale factor of the spacetimes. Finally, we provide some physical models that support our general setup.
\end{abstract}




\section{Introduction}
It is of interest to study the Klein-Gordon equation because it describes the dynamics of the spinless particle in our universe at quantum level or it can be viewed as the scalar wave equation at classical level in our universe. Moreover, to get a more realistic picture we have to include the influence of geometrical properties of the universe which may also determine by its matter distribution. It is extremely difficult, however, to solve the Klein-Gordon equation on the family of four dimensional Friedmann-Robertson-Walker spacetimes with general couplings even at the classical level. Therefore, we have to specify both the spacetime and the coupling in order to get a solvable model. Several simple models in four dimensions have been studied, for example, in  \cite{KlainermanSarnak:1981, AbbasiCraig:2014, Yagdijan:2013}. 

In this paper we prove the local and global existence of solutions of the Klein-Gordon equation in higher dimensional spatially flat Friedmann-Robertson-Walker spacetimes with non-minimal coupling between the scalar curvature and the scalar field, and the scalar potential turned on. This additional non-minimal coupling is the simplest generalization of the scalar field theory on curved spacetimes \cite{Callan:1970,Birrell:1980}, which can be viewed, for example, as a result of quantum corrections \cite{Ishikawa:1983, Birrelbooks:1984}.

The starting point of proving  the local and global existence of solutions is by introducing the $k$-order linear energy and the corresponding energy norm. In the local case, we show that in order to admit a regular solution both energy functionals must be bounded below a real constant $C$ for finite time $T < \infty$ with initial data in $H^{k+1}\times H^{k}$. As we take $T \to \infty$, namely the global existence, we give three additional assumptions related to the nonminimal coupling constant, the scale factor of spacetimes, and the form of the scalar that has to be a polynomial with small positive parameter. Using these assumptions, we prove that the supremum of the energy norm is bounded and thus, we could have a decay estimate.

We organize the paper as follows. In Section \ref{sec:SpatiallyFlatSpacetimes} we briefly review spatially flat spacetimes in higher dimension. We discuss some local properties of the real scalar field on higher dimensional spatially flat spacetimes by introducing $k$-order linear energy and the corresponding energy norm in Section \ref{sec:RealScalarSpatiallyFlatSpacetimes}. In Section \ref{sec:LocalEx} we provide a proof of the local existence and the uniqueness of solutions together with a smoothness property. In Section \ref{sec:GlobalEx} we prove that the solutions could exist globally and they have a particular decay estimate. Finally, we discuss some models, in which the global solution does exist in Section \ref{sec:Models}.




\section{Spatially Flat Spacetimes in Higher Dimension}
\label{sec:SpatiallyFlatSpacetimes}

In this section, we shortly review the higher dimensional conformally flat  spacetime which can be constructed by the   $D$-dimensional spatially flat Lorentzian manifold, $\mathcal{M}^{D}$, $D \geq 4$ with standard coordinates, $x^\mu = (x^{0}=t,x^{i})$, $\mu = 0,1,\dots,D-1$, $i=1,2,\dots,D-1$, and is equipped by Lorentzian metric with signature $\{-1,1,\dots,1\}$. We can write down the metric  as
\begin{equation}
ds^2 = -dt^2 + a^2(t)\sum_{i=1}^{D-1}dx_i^2 \:, \label{fh2}
\end{equation}
with $x^{i}$ being the usual Cartesian coordinates for $\lR^{D-1}$. Defining a new time coordinate $\tau$ by
\begin{equation}
\frac{d\tau}{dt}= \frac{1}{a(t)}\:, \label{ja}
\end{equation} 
 we can write (\ref{fh2}) as
\begin{equation}
ds^2  = a^2(\tau) \left(-d\tau^2 + \sum_{i=1}^{D-1}dx_i^2 \right)\:. \label{ea}
\end{equation}
Thus, our spacetime $\mathcal{M}^{D}$ is conformal to flat Minkowski space $M^{D+1} \simeq \lR \times \lR^{D-1}$. In terms of components, we can write the metric of $\mathcal{M}^{D}$ as
\begin{equation}
g_{\mu\nu} = a^2(\tau)\eta_{\mu\nu}\:, \label{spacetimemetric}
\end{equation}
where $\eta_{\mu\nu} = \mathrm{diag}(-1,1,\dots,1) $ is the component of Minkowski metric. Furthermore, it is of interest to write down the Ricci tensor and the scalar curvature related to metric (\ref{spacetimemetric}) 
\begin{eqnarray}
R_{\mu\nu} &=& \dot{H} \left(\eta_{\mu\nu} - (D-2) \delta^{0}_{\mu} \;\delta^{0}_{\nu}\right) + H^2 (D-2) \left(\eta_{\mu\nu} +  \delta^{0}_{\mu} \;\delta^{0}_{\nu}\right) \:,\nonumber\\
R &=& (D-1) a^{-2} \left( 2  \dot{H} + (D-2) H^2\right) \:, \label{scalarcurv}
\end{eqnarray}
respectively, assuming that the scale factor $a(\tau)$ belongs to $C^n$-function with $n \ge 2$ for all $\tau > 0$ where $\dot{a} \equiv da/d\tau$ and we have defined the Hubble parameter $H \equiv \dot{a}/ a$.

The higher dimensional Friedmann equations describing an accelerated universe for single component matter are given by \cite{Chatterjee:1990tq}
\begin{eqnarray}
\frac{(D-2)(D-1)}{2a^4}H^2 & = & 8\pi\rho ~ ,\nonumber \\
-\frac{(D-2)}{a^3}\dot{H} - \frac{(D-2)(D-5)}{2a^4}H^2 & = & 8\pi P\:,
\end{eqnarray}
together with equation of state $P=w\rho$. The general solutions of this equation have the form
\begin{equation}\label{gensola}
a(\tau) = \begin{cases}
(\tau+\tau_{0})^{\frac{2}{(D-1)(w+1) - 2}} ~ , & w \neq -(D-3)/(D-1)\\
 e^{\alpha (\tau + \tau_{0})} ~ , & w = -(D-3)/(D-1)
\end{cases}
\end{equation}
which will be useful for our analysis in the last part of this paper. Some examples, which may be considered as the higher dimensional standard model of cosmology, are listed in Table \ref{tab1} \cite{Chatterjee:1990tq}.

\begin{table}
	\centering
	\begin{tabular}{| l | l | l |}
		\hline & &  \\
		Conditions & $w$ & $a(\tau)$  \\
		\hline \hline & & \\
		 Matter Dominated & $w = 0$ & $(\tau+\tau_0)^{\frac{2}{D-3}}$ \\
		\hline  & & \\
		$\Lambda$-Dominated & $w = -1$ & $(\tau+\tau_0)^{-1}$\\
		\hline  & &\\
		Radiation Dominated & $w = \frac{1}{D-1}$ & $(\tau+\tau_0)$ \\
		\hline
		\end{tabular}
		\\
		\vspace{1.5mm}
	\caption{Standard Models of Cosmology in higher dimension.}
	\label{tab1}
\end{table}

For a matter dominated universe, which is called higher dimensional Einstein-de Sitter universe, the universe consists only of non-relativistic matter (dust) and has zero cosmological constant. The Ricci tensor related to metric (\ref{spacetimemetric}) is given by
\begin{equation}
R_{\mu\nu} =\frac{2 (D-1)}{(D-3)^2 \tau^2} \left(\eta_{\mu\nu} + (D-2) \delta^{0}_{\mu} \;\delta^{0}_{\nu}\right) \:,
\end{equation}
and the scalar curvature is given by
\begin{equation}
R =  \frac{4(D-1)}{(D-3)^2\tau^{2(D-1)/(D-3)}}\:. \label{scalarcurv1}
\end{equation}

\section{Real Scalar Field in Spatially Flat Universe}
\label{sec:RealScalarSpatiallyFlatSpacetimes}

In this section, we discuss some local properties of the real scalar field on the higher dimensional spatially flat universe as a background with additional non-minimal coupling where the coupling interaction of the scalar field $\phi$ is proportional to the scalar curvature of the spacetime. We show that the nonlinear terms and the $k$-order linear energy are bounded if the energy norm is also bounded.

The action of our theory  has the form
\begin{equation}
\mathcal{S} = \int d\tau dx \sqrt{-g} \left(\frac{1}{2}\partial_\mu\phi\:\partial^\mu\phi+\frac{\xi}{2}R\phi^2 - V(\phi)\right)\:, \label{action}
\end{equation} 
where $g$ and $R$ are the determinant and the scalar curvature of the metric (\ref{spacetimemetric}), respectively. The non-minimal coupling is introduced in the second term of the right hand side of (\ref{action}) with positive constant $\xi$ \footnote{Model with positive $\xi$ is called canonical, while that with negative $\xi$ is called phantom.}. The real function $V(\phi)$ denotes the scalar potential which is assumed to be smooth and satisfy the following conditions,
\begin{eqnarray}
(1) && V(0) = 0 ~ ,\nonumber\\
(2) && \partial_{\phi}V(0) = 0 ~ .\label{Vcond}
\end{eqnarray}
%
The conditions in (\ref{Vcond}) are satisfied by several known scalar potentials such as the $\phi^4$ theory and the sine-Gordon theory.

The equation of motions of scalar field in this case is given by
\begin{equation}
\nabla_\mu\nabla^\mu\phi-\xi R\phi + \partial_\phi V(\phi) = 0\:,
\end{equation}
where $\nabla_\mu$ is a covariant derivative with respect to the metric (\ref{ea}). Using (\ref{ea}) and (\ref{scalarcurv}), we can write the equation of motion in form of nonlinear waves equation
\begin{equation}
\partial_{\tau}^2 \phi  - \Delta\phi = F(\phi,\partial_\tau\phi) \label{iPDE} \:,
\end{equation}
where
\begin{equation}
F(\phi,\partial_\tau\phi) \equiv -(D-2) H~ \partial_\tau\phi- \xi (D-1)\left( 2 \dot{H} + (D-2) H^2\right)\phi + a^2 \partial_{\phi}V(\phi) \:, \label{nonlinearterm}
\end{equation}   
and $\Delta$ is Laplacian in $\lR^{D-1}$. The energy-momentum tensor for system (\ref{action}) is given by
\begin{eqnarray}
T_{\mu\nu}  =  \nabla_\mu\phi\;\nabla_\nu\phi - \frac{1}{2}g_{\mu\nu} \nabla_\lambda\phi\;\nabla^\lambda\phi - g_{\mu\nu} V(\phi) - \xi G_{\mu\nu}\phi^2 -\xi\left(g_{\mu\nu} \nabla_\lambda\nabla^\lambda\phi^2 - \nabla_\mu\nabla_\nu\phi^2 \right)\:,\nonumber\\
\end{eqnarray}
where $G_{\mu\nu}$ is Einstein tensor of $\mathcal{M}^D$. Then, we could define an energy functional $E = a^{2D-2}\int T_{00} ~ d^{D-1}x$. However, such a functional covers only $L^2$-norm functions and it is difficult to obtain a decay estimate of the scalar field from it in the global case (see Section \ref{sec:GlobalEx}) since it contains the nonlinear terms, namely the scalar potential and the nonminimal coupling. To overcome the problem, we introduce the $k$-order linear energy 
\begin{equation}
\mathcal{H}_{k}[\phi] = \frac{1}{2}\sum_{|\alpha|\leq k}\int_{\lR^{D-1}}\:\left[(\partial^\alpha\partial_\tau \phi)^2 + \left|\nabla\partial^\alpha \phi\right|^2\right]\:dx\:, \label{energy}
\end{equation}
with $\alpha$ being a multi index, which will be used in this paper. We also define an energy norm as
\begin{equation}
|\phi(\tau,\cdot)|_{k} := \|\phi(\tau,\cdot)\|_{H^{k+1}(\lR^{D-1})} + \|\partial_{\tau}\phi(\tau,\cdot)\|_{H^{k}(\lR^{D-1})}
\end{equation}
This form of energy can be used as a bound for the nonlinear term. In the following lemmas, we prove the properties of the nonlinear term and the linear energy.

\begin{lemma}
	\label{nonlinearbound}
	Let $\phi$ be a real function such that for all $\tau \in [\tau_0,\tau_0 + T]$ and $k\in\mathbb{N}_{0}$,
	\begin{equation}
	|\phi(\tau,\cdot)|_{k} \leq C\:.
	\end{equation}
	Then for all $\tau \in [\tau_0,\tau_0 + T]$ and $k>(D-1)/2$, we have
	\begin{equation}
	\left(\sum_{|\alpha|\leq k}^{} \int_{\lR^{D-1}} [\partial^\alpha F(\phi,\partial_\tau \phi)]^2 \:dx\right)^{1/2} \leq C \:, \label{nonlinearboundeq}
	\end{equation}
	where $C$ depends only on the initial data, $T$, $k$ and the bound of the scalar potential.
\end{lemma}

\begin{proof}
	First, we consider the case of $|\alpha|>0$. The spatial derivative of Equation (\ref{nonlinearterm}) gives
	\begin{equation}
	\partial^\alpha F(\phi,\partial_\tau\phi) = -(D-2) H ~ \partial^\alpha \partial_\tau\phi- \xi (D-1)\left( 2 \dot{H} + (D-2) H^2\right)\partial^\alpha\phi + a^2\partial^\alpha \partial_{\phi} V\:.
	\end{equation}
	Since $\tau \in [\tau_0,\tau_0 + T]$, $a(\tau)$ is a regular $C^n$-function, and using the hypothesis that the first and second terms are bounded, we thus have 
	\begin{equation}
	\left(\sum_{|\alpha|\leq k}^{} \int_{\lR^{D-1}} [\partial^\alpha F(\phi,\partial_\tau \phi)]^2 \:dx\right)^{1/2} \leq C\left\{1 + \|\partial_{\phi} V\|_{H^{k}(\lR^{D-1})} \right\}
	\end{equation}

	To estimate the derivative of the scalar potential, we write down $\partial^\alpha \partial_{\phi} V$ as
	\begin{equation}
	\partial^\alpha \partial_{\phi} V = \left(\partial_{\phi}^{\beta} \partial_{\phi} V\right)\;\partial^{\gamma_1}\phi\;\partial^{\gamma_2}\phi\ldots\partial^{\gamma_i}\phi\:,
	\end{equation}
	where $\gamma_1 + \gamma_2 + \ldots + \gamma_i = \alpha$. Since the scalar potential is a smooth function, then by Sobolev embedding theorem for $k>(D-1)/2$, the hypothesis of the lemma implies that $\partial_{\phi}^{\beta} \partial_{\phi} V$ is bounded. Thus,
	\begin{equation}
	\|\partial_{\phi} V\|_{H^{k}(\lR^{D-1})} \leq C\|\phi\|_{H^{k}(\lR^{D-1})} \leq C\:,
	\end{equation}
	and we obtain Equation (\ref{nonlinearboundeq}).
	
	For $|\alpha| = 0$, we have
	\begin{equation}
	\left(\int_{\lR^{D-1}}\:[F(\phi,\partial_\tau \phi)]^2 \:dx\right)^{1/2} \leq C \left\{\|\phi\|_{L^{2}(\lR^{D-1})} + \|\partial_{\tau} \phi\|_{L^{2}(\lR^{D-1})} + \|\partial_{\phi} V\|_{L^{2}(\lR^{D-1})}\right\}\:.
	\end{equation}
	The first and second terms at the right hand side are bounded to a constant by the hypothesis. Using the assumption that $\partial_{\phi}V(0) = 0$, we obtain the estimate
	\begin{equation}
	\|\partial_{\phi} V\|_{L^{2}(\lR^{D-1})} \leq C \|\phi\|_{L^{2}(\lR^{D-1})} \leq C\:,
	\end{equation}
	where the constant $C$ depends on the bound of the scalar potential. Hence, the proof is complete. 
\end{proof}

\begin{lemma}
	\label{energybound}
	Let $\phi$ be a real function such that for all $\tau \in [\tau_0,\tau_0 + T]$ and $k\in\mathbb{N}_{0}$,
	\begin{equation}
	|\phi(\tau,\cdot)|_{k} \leq C\:.
	\end{equation} 
	Then for $\tau \in [\tau_0,\tau_0 + T]$ and $k>(D-1)/2$, we have
	\begin{equation}
	\mathcal{H}_k^{1/2}[\phi](\tau) \leq \mathcal{H}_k^{1/2}[\phi](\tau_0) + \frac{1}{2}CT\:, 
	\end{equation}
	where the constant $C$ depends only on the initial data, $T$, $k$ and the bound of the scalar potential.
\end{lemma}

\begin{proof}
	Let us consider,
	\begin{eqnarray}
	\frac{d\mathcal{H}_k[\phi]}{d\tau} & = & \frac{1}{2}\sum_{|\alpha|\leq k}\int_{\lR^{D-1}}\: \left[2(\partial^\alpha\partial_\tau\phi)(\partial^\alpha\partial_\tau^2\phi)+2(\nabla\partial^\alpha \phi)\cdot(\nabla\partial^\alpha\partial_\tau\phi)\right]dx\nonumber\\
	& = &  \sum_{|\alpha|\leq k} \int_{\lR^{D-1}} \partial^\alpha F(\phi,\partial_\tau \phi)\; \partial^\alpha\partial_\tau \phi \:dx \:.
	\end{eqnarray}
	Using Schwartz and H\"older inequalities, we obtain
	\begin{equation}
	\frac{d\mathcal{H}_k[\phi]}{d\tau} \leq \left(\sum_{|\alpha|\leq k}^{} \int_{\lR^{D-1}} [\partial^\alpha F(\phi,\partial_\tau \phi)]^2 dx\right)^{1/2}\left(\sum_{|\alpha|\leq k}^{} \int_{\lR^{D-1}} (\partial^\alpha\partial_\tau \phi)^2dx\right)^{1/2}\:.
	\end{equation}
	The second factor can be estimated by a constant times $\mathcal{H}_k^{1/2}[\phi]$. Since $\phi$ satisfies the hypothesis of Lemma \ref{nonlinearbound}, then the first factor is bounded to a constant. Thus, we obtain
	\begin{equation}
	\frac{d\mathcal{H}_k[\phi]}{d\tau} \leq C \mathcal{H}_k^{1/2}[\phi]\:.
	\end{equation}
	Since $\mathcal{H}_k[\phi]$ is positive, we can divide the inequality by $\mathcal{H}_k^{1/2}[\phi]$ and integrate to obtain
	\begin{equation}
	\mathcal{H}_k^{1/2}[\phi](\tau) \leq \mathcal{H}_k^{1/2}[\phi](\tau_0) + \frac{1}{2}CT\:, 
	\end{equation}
	and the proof is finished.
\end{proof}

A classical solution of a scalar field on spatially flat spacetimes with non-minimal coupling is a real smooth function $\phi$ satisfying Equation (\ref{iPDE}). Then, a generalized solution is a real function $\phi \in C^0\left([\tau_0,\tau_0 + T], H^{k+1}(\lR^{D-1})\right) \cap C^1 \left([\tau_0,\tau_0 + T], H^{k}(\lR^{D-1})\right)$ such that Equation (\ref{iPDE}) is satisfied in a distributional sense. In the rest of this paper, we will prove the existence and uniqueness of both generalized and classical solutions to Equation (\ref{iPDE}) together with the initial data
\begin{eqnarray}
\phi(\tau_{0},x) & = & f(x)\:, \nonumber\\
\partial_{\tau}\phi(\tau_{0},x) & = & g(x)\:,
\end{eqnarray} 
with $f \in H^{k+1}(\lR^{D-1})$ and $g \in H^{k}(\lR^{D-1})$ and having compact support.

\section{Local Existence, Uniqueness, and Smoothness}
\label{sec:LocalEx}

In this section, we prove the local existence and the uniqueness of the solution of the scalar field Equation (\ref{iPDE}).

\subsection{Local Existence and Uniqueness}
The scalar field equation in spatially flat universe is given by,
\begin{equation}
\begin{cases}
\partial_{\tau}^2 \phi  - \Delta\phi = F(\phi,\partial_\tau\phi) \\
\phi(\tau_{0},\cdot)  =  f \in H^{k+1}(\lR^{D-1}) \\
\partial_{\tau}\phi(\tau_{0},\cdot)  =  g \in H^{k}(\lR^{D-1}) \:,
\end{cases} \label{MainEq}
\end{equation}
where the nonlinear term is given by Equation (\ref{nonlinearterm}). Let us consider the sequence $\{\phi_{l}\}$ such that
\begin{equation}
\begin{cases}
\partial_{\tau}^2 \phi_{0}  - \Delta\phi_{0} = 0 \\
\phi_{0}(\tau_{0},\cdot)  =  f_{0} \\
\partial_{\tau}\phi_0(\tau_{0},\cdot)  =  g_{0} \:,
\end{cases}
\label{zerotheq}
\end{equation}
and for $l\geq 0$,
\begin{equation}
\begin{cases}
\partial_{\tau}^2 \phi_{l+1}  - \Delta\phi_{l+1} = F(\phi_{l},\partial_\tau\phi_{l}) \\
\phi_{l}(\tau_{0},\cdot)  =  f_{l} \\
\partial_{\tau}\phi_{l}(\tau_{0},\cdot)  =  g_{l} \:.
\end{cases}
\label{higheq}
\end{equation}
Since the Schwartz space $\mathcal{S}(\lR^{D-1})$ is dense in $H^{k}(\lR^{D-1})$, then we can choose the sequences $\{f_{l}\}$ and $\{g_{l}\}$ such that $f_l,g_l \in \mathcal{S}(\lR^{D-1})$ and $f_l \rightarrow f$ in $H^{k+1}(\lR^{D-1})$ and $g_l \rightarrow g$ in $H^{k}(\lR^{D-1})$. Without loss of generality, we can choose $f_{l}$ and $g_{l}$ such that,
\begin{eqnarray}
\|f_{l}\|_{H^{k+1}(\lR^{D-1})} & \leq & 2\|f\|_{H^{k+1}(\lR^{D-1})} \nonumber\\
\|g_{l}\|_{H^{k}(\lR^{D-1})} & \leq & 2\|g\|_{H^{k}(\lR^{D-1})}\:.
\end{eqnarray}
In the following lemma, we prove that $\mathcal{H}_{k}[\phi_l]$ is bounded for all $l$. 
\begin{lemma}
	\label{iterationboundenergy}
	Let $\{\phi_l\}$ be solutions of (\ref{zerotheq}) and (\ref{higheq}). Let $\mathcal{H}_{k}$ be the linear energy defined in Equation (\ref{energy}). For $k>(D-1)/2$, there exist constants $C$ and $T$ such that,
	\begin{equation}
	\mathcal{H}_{k}[\phi_l](\tau) \leq C\:, \label{iterationenergybound}
	\end{equation} 
	for all $l \geq 0$ and $\tau \in [\tau_0,\tau_0 + T]$ . The constant $C$ depends on the initial data, $k$ and the bound of the scalar potential.
\end{lemma}

\begin{proof}
	We will prove the lemma by induction. Since $\phi_{0}$ is a solution of linear wave equation, then $\mathcal{H}_{k}$ is conserved. Thus,
	\begin{equation}
	\mathcal{H}_{k}[\phi_0](\tau_0) = \mathcal{H}_{k}[\phi_0](\tau) \leq \tilde{C}\:, \label{iterationenergyinitial}
	\end{equation}
	where the bound constant $\tilde{C}$ depends only on the initial data such that,
	\begin{equation}
	\mathcal{H}_{k}[\phi_l](\tau_0) \leq \tilde{C}\:.
	\end{equation}
	Hence, Equation (\ref{iterationenergybound}) is satisfied for $l=0$.	
	
	Now, we assume  Equation (\ref{iterationenergybound}) to be true for $l=n$. Then, we obtain
	\begin{eqnarray}
	\frac{\partial}{\partial \tau}\int_{\lR^{D-1}}\:|\phi_{n}|^2\:dx & = & 2 \int_{\lR^{D-1}}\:\phi_{n}\;\partial_{\tau}\phi_{n}\:dx \nonumber \\
	& \leq & 2 \|\phi_{n}\|_{L^{2}(\lR^{D-1})} \|\partial_{\tau}\phi_{n}\|_{L^{2}(\lR^{D-1})} \leq C \|\phi_{n}\|_{L^{2}(\lR^{D-1})}\mathcal{H}^{1/2}_{k}[\phi_n]\:,
	\end{eqnarray}
	where we have  used H\"older's inequality and the definition of $\mathcal{H}_{k}$ in Equation (\ref{energy}). Integrating the inequality and using the induction hypothesis, we obtain
	\begin{equation}
	\|\phi_{n}(\tau,\cdot)\|_{L^{2}(\lR^{D-1})} \leq C \left( \|f\|_{L^{2}(\lR^{D-1})} + T \right)\:.
	\end{equation}
	If we assume $T \leq 1$, then $\|\phi_{n}(\tau,\cdot)\|_{L^{2}(\lR^{D-1})}$ is bounded for $\tau \in [\tau_0,\tau_0 + T]$, which depends only on the initial data.
	
	Using Sobolev embedding theorem for $k>(D-1)/2$, we have bounds on $\sum_{j=1}^{D-1} \|\partial_j\phi_l(\tau,\cdot)\|_{C_b(\lR^{D-1})}$ and $\|\partial_t\phi_l(\tau,\cdot)\|_{C_b(\lR^{D-1})}$. Furthermore, we also have
	\begin{equation}
	|\phi_n(\tau,\cdot)|_{k} \leq C\:, \label{phinbound}
	\end{equation}
	where the constant only depends on initial data. Hence, $\phi_n$ satisfies the hypothesis of Lemmas \ref{nonlinearbound} and \ref{energybound}.
	
	For $l=n+1$, we have
	\begin{eqnarray}
	\frac{d\mathcal{H}_k[\phi_{n+1}]}{d\tau} & = &  \sum_{|\alpha|\leq k} \int_{\lR^{D-1}} \partial^\alpha F(\phi_n,\partial_\tau \phi_n)\; \partial^\alpha\partial_\tau \phi_{n+1} \:dx \nonumber \\
	& \leq &  \left(\sum_{|\alpha|\leq k} \int_{\lR^{D-1}} [\partial^\alpha F(\phi_n,\partial_\tau \phi_n)]^2 dx\right)^{1/2}\left(\sum_{|\alpha|\leq k} \int_{\lR^{D-1}} (\partial^\alpha\partial_\tau \phi_{n+1})^2dx\right)^{1/2} \nonumber\\
	& \leq & C \mathcal{H}^{1/2}_k[\phi_{n+1}]\:,
	\end{eqnarray}
	where we used Lemma \ref{nonlinearbound} in the last inequality. Integrating the inequality, we obtain
	\begin{equation}
	\mathcal{H}^{1/2}_k[\phi_{n+1}](\tau) \leq \mathcal{H}^{1/2}_k[\phi_{n+1}](\tau_0) + C T\:.
	\end{equation}
	Using Equation (\ref{iterationenergyinitial}) and assuming $T \leq 1$, then $\mathcal{H}_k[\phi_{n+1}](\tau)$ is bounded for $\tau \in [\tau_0,\tau_0 + T]$.

\end{proof}

Let us define 
\begin{equation}
E_{l,k} = \sup_{\tau \in [\tau_0,\tau_0 + T]}\left[\mathcal{H}^{1/2}_{k}[\phi_l - \phi_{l-1}](\tau) + \|(\phi_{l}-\phi_{l-1})(\tau,\cdot)\|_{L^{2}(\lR^{D-1})} \right]\:. \label{El}
\end{equation}
Next, we derive an estimate for the difference of consecutive sequence which is important to prove the convergence of the sequence.

\begin{lemma}
	\label{energyestimate}
	Let $\{\phi_l\}$ be solutions of (\ref{zerotheq}) and (\ref{higheq}). Let $\mathcal{H}_{k}$ is the linear energy defined in Equation (\ref{energy}). Then, for $k>(D-1)/2$, we have the estimate
	\begin{equation}
	\mathcal{H}^{1/2}_{k}[\phi_l - \phi_{l-1}](\tau) \leq \mathcal{H}^{1/2}_{k}[\phi_l - \phi_{l-1}](\tau_0) + \frac{1}{2} C E_{l,k} T\:,
	\end{equation}
	for all $\tau \in [\tau_0,\tau_0 + T]$.
\end{lemma}

\begin{proof}
	First, we have the estimate
	\begin{eqnarray}
	\|F(\phi_l,\partial_{\phi}\tau \phi_l) - F(\phi_{l-1},\partial_{\phi}\tau \phi_{l-1})\|_{H^{k}(\lR^{D-1})} & \leq & C\left[\|\phi_l - \phi_{l-1}\|_{H^{k}(\lR^{D-1})} + \|\partial_{\tau}\phi_l - \partial_{\tau}\phi_{l-1}\|_{H^{k}(\lR^{D-1})} \right. \nonumber \\
	& & \qquad \left. + \|\partial_{\phi}V(\phi_l) - \partial_{\phi}V(\phi_{l-1})\|_{H^{k}(\lR^{D-1})} \right]\:.
	\end{eqnarray}
	We can estimate the first and second terms as
	\begin{eqnarray}
	\|\phi_l - \phi_{l-1}\|_{H^{k}(\lR^{D-1})} & \leq & \|\phi_l - \phi_{l-1}\|_{L^{2}(\lR^{D-1})} + \mathcal{H}^{1/2}_{k}[\phi_l - \phi_{l-1}]\nonumber \\
	\|\partial_{\tau}\phi_l - \partial_{\tau}\phi_{l-1}\|_{H^{k}(\lR^{D-1})} & \leq & \mathcal{H}^{1/2}_{k}[\phi_l - \phi_{l-1}]\:.
	\end{eqnarray}
	To estimate the third term, we write
	\begin{equation}
	\partial_\phi V(\phi_l) - \partial_\phi V(\phi_{l-1}) = \int_{0}^{1} \partial_\phi^2 V[\sigma\phi_l + (1-\sigma)\phi_{l-1}]d\sigma \;\left(\phi_{l}-\phi_{l-1}\right)\:.
	\end{equation}
	 Using Lemma \ref{iterationboundenergy} and the fact that the scalar potential is a smooth function, for $k>(D-1)/2$ we obtain
	\begin{eqnarray}
	\|\partial_{\phi}V(\phi_l) - \partial_{\phi}V(\phi_{l-1})\|_{H^{k}(\lR^{D-1})} & \leq &  C \|\phi_{l}-\phi_{l-1}\|_{H^{k}(\lR^{D-1})} \nonumber\\
	& \leq & C\left(\|\phi_{l}-\phi_{l-1}\|_{L^{2}(\lR^{D-1})} + \mathcal{H}^{1/2}_{k}[\phi_l - \phi_{l-1}]\right)\:,\nonumber\\
	\end{eqnarray}
	where the constant $C$ depends on the bound of the scalar potential. Hence, we obtain the estimate
	\begin{equation}
	\|F(\phi_l,\partial_{\phi}\tau \phi_l) - F(\phi_{l-1},\partial_{\phi}\tau \phi_{l-1})\|_{H^{k}(\lR^{D-1})} \leq C\left(\|\phi_{l}-\phi_{l-1}\|_{L^{2}(\lR^{D-1})} + \mathcal{H}^{1/2}_{k}[\phi_l - \phi_{l-1}]\right) \leq C E_{l,k}\:.
	\end{equation}
	
	Now, similar to the proof of Lemma \ref{iterationboundenergy}, we have
	\begin{eqnarray}
	\frac{d}{d\tau}\mathcal{H}_k[\phi_{l}-\phi_{l-1}] & \leq & C \|F(\phi_l,\partial_\tau \phi_l) - F(\phi_{l-1},\partial_\tau \phi_{l-1})\|_{H^{k}(\lR^{D-1})}\;\|\partial_{\tau}\phi_l - \partial_{\tau}\phi_{l-1}\|_{H^{k}(\lR^{D-1})} \nonumber \\
	& \leq & CE_{l,k} \mathcal{H}^{1/2}_{k}[\phi_l - \phi_{l-1}]\:.
	\end{eqnarray}
	Integrating the inequality, we obtain
	\begin{equation}
	\mathcal{H}^{1/2}_{k}[\phi_l - \phi_{l-1}](\tau) \leq \mathcal{H}^{1/2}_{k}[\phi_l - \phi_{l-1}](\tau_0) + \frac{1}{2} C E_{l,k} T\:,
	\end{equation}
	for all $\tau \in [\tau_0,\tau_0 + T]$ and the proof is finished.
	
\end{proof}

Next, we prove the estimate of $E_{l,k}$.

\begin{lemma}
	\label{lemmaboundEl}
	Let $\{\phi_l\}$ be solutions of (\ref{zerotheq}) and (\ref{higheq}). Let $E_{l,k}$ be given by Equation (\ref{El}). For $k>(D-1)/2$, there exist constants $C_0>1$ and $T$ such that,
	\begin{equation}
	E_{l,k} \leq \frac{C_0}{2^{l}}\:, \label{boundEl}
	\end{equation} 
	for all $l$.
\end{lemma} 

\begin{proof}
	The lemma is true for $l=1$ by assuming that the constant $C_{0}$ is big enough. Let us consider
	\begin{eqnarray}
	\frac{\partial}{\partial \tau}\int_{\lR^{D-1}}\:|\phi_{l+1}-\phi_{l}|^2\:dx & = & 2 \int_{\lR^{D-1}}\:|\phi_{l+1}-\phi_{l}|\;\partial_{\tau}|\phi_{l+1}-\phi_{l}|\:dx \nonumber \\
	& \leq & 2 \|\phi_{l+1}-\phi_{l}\|_{L^{2}(\lR^{D-1})} \|\partial_{\tau}(\phi_{l+1}-\phi_{l})\|_{L^{2}(\lR^{D-1})} \nonumber \\ 
	& \leq & 2^{3/2} \|\phi_{l+1}-\phi_{l}\|_{L^{2}(\lR^{D-1})}\mathcal{H}^{1/2}_{k}[\phi_{l+1}-\phi_{l}]\:,
	\end{eqnarray}
	Integrating this inequality and since $k>(D-1)/2$, we obtain
	\begin{equation}
	\|\phi_{l+1}-\phi_{l}\|_{L^{2}(\lR^{D-1})}(\tau) \leq \|\phi_{l+1}-\phi_{l}\|_{L^{2}(\lR^{D-1})}(\tau_{0}) + 2^{1/2}\left|\int_{\tau_{0}}^{\tau}\mathcal{H}^{1/2}_{k}[\phi_{l+1}-\phi_{l}](s)\:ds\right|\:.
	\end{equation}
	Assuming $T<1/2$ and using Lemma \ref{energyestimate}, we have the estimate
	\begin{equation}
	E_{l+1,k} \leq \|\phi_{l+1}-\phi_{l}\|_{L^{2}(\lR^{D-1})}(\tau_{0}) + 2 \mathcal{H}^{1/2}_{k}[\phi_{l+1}-\phi_{l}](\tau_{0}) + C E_{l,k}T\:. \label{El+1}
	\end{equation}
	The first and second terms depend only on the initial data. However, we can choose them as such that,
	\begin{equation}
	\|\phi_{l+1}-\phi_{l}\|_{L^{2}(\lR^{D-1})}(\tau_{0}) + 2 \mathcal{H}^{1/2}_{k}[\phi_{l+1}-\phi_{l}](\tau_{0}) \leq \frac{1}{2^{l+2}}\:.
	\end{equation}
	Hence, by assuming $CT<1/4$, Equation (\ref{El+1}) and the induction hypothesis give
	\begin{equation}
	E_{l+1,k} \leq \frac{1}{2^{l+2}} + \frac{C_{0}}{2^{l+2}} \leq \frac{C_{0}}{2^{l+1}}\:.
	\end{equation}
	Thus, Equation (\ref{boundEl}) is true for all $l$ and the proof is finished.
	
\end{proof}

Let us consider, for all $|\alpha| \leq k+1$, 
\begin{eqnarray}
\|\nabla\partial^\alpha (\phi_{l}-\phi_{l-1})(\tau,\cdot)\|_{L^2(\lR^{D-1})} & = & \left(\int_{\lR^{D-1}}^{}\left|\nabla\partial^\alpha (\phi_{l}-\phi_{l-1})(\tau,\cdot)\right|^2 dx\right)^{1/2} \nonumber\\
& \leq & C \mathcal{H}_{k}[\phi_{l}-\phi_{l-1}]^{1/2}(\tau) \nonumber\\
& \leq & C E_{l,k}\:.
\end{eqnarray}
In other words, $\phi_{l}-\phi_{l-1} \in H^{k+1}(\lR^{D-1})$. However, using Lemma \ref{lemmaboundEl}, we have
\begin{equation}
\sup_{\tau \in [\tau_0,\tau_0 + T]} \|\phi_{l}-\phi_{l-1}\|_{H^{k+1}(\lR^{D-1})} \leq \frac{C_{0}}{2^{l}}\:,
\end{equation}
which show that $\{\phi_{l}\}$ is a Cauchy sequence on $C\left([\tau_0,\tau_0 + T], H^{k+1}(\lR^{D-1})\right)$. Using a similar method, we can show that $\{\partial_\tau \phi_l\}$ is also a Cauchy sequence on $C\left([\tau_0,\tau_0 + T],H^{k}(\lR^{D-1})\right)$. Thus, we have proven the existence of generalized solutions of Equation (\ref{iPDE}) such that
\begin{equation}
\phi \in C\left([\tau_0,\tau_0 + T], H^{k+1}(\lR^{D-1})\right) \cap C^1 \left([\tau_0,\tau_0 + T], H^{k}(\lR^{D-1})\right)\:.
\end{equation}

Furthermore, using Lemma \ref{lemmaboundEl}, we also get that $\{\partial_i\partial_j \phi_l\}$ and $\{\partial_i\partial_\tau \phi_l\}$ are Cauchy sequences on $C\left([\tau_0,\tau_0 + T],H^{k-1}(\lR^{D-1})\right)$. In fact, we have
\begin{eqnarray}
\|\partial_\tau^2(\phi_{l+1}-\phi_l)\|_{H^{k-1}(\lR^{D-1})} & \leq & \|F(\phi_{l},\partial_\tau\phi_{l})-F(\phi_{l-1},\partial_\tau\phi_{l-1}) \|_{H^{k-1}(\lR^{D-1})} \nonumber\\
& &\qquad + \|\partial_i\partial^i(\phi_{l+1}-\phi_l)\|_{H^{k-1}(\lR^{D-1})}\:.
\end{eqnarray}
On the right hand side, the first term is bounded by $E_{l,k}$ and the second term is bounded by a constant since it is a Cauchy sequence on $H^{k-1}(\lR^{D-1})$. Then using Lemma \ref{lemmaboundEl} , we have that $\{\partial_\tau^2\phi_{l}\}$ is also a Cauchy sequence in $H^{k-1}(\lR^{D-1})$. Thus, for $k>(D-1)/2$, there exists a real function $\phi$ such that,
\begin{equation}
\sup_{\tau \in [\tau_0,\tau_0 + T]}\|\partial^{2}_{\tau}\phi\|_{H^{k-1}(\lR^{D-1})} \leq C\:.
\end{equation}
The above inequality implies that $\partial_\tau^2\phi \in C\left([\tau_0,\tau_0 + T],H^{k-1}(\lR^{D-1})\right)$.

Let us consider a fixed point $(\tau,x) \in \mathcal{M}^D$. We make a sequance $(t_{l},x_{l}) \rightarrow (\tau,x)$, where $\tau_0 \leq \tau_{l},\tau \leq \tau_{0}+T$. Now, we have the estimate
\begin{equation}
\left|\partial_\tau^2\phi(\tau,x)-\partial_\tau^2\phi(\tau_{l},x_{l})\right| \leq \left|\partial_\tau^2\phi(\tau,x)-\partial_\tau^2\phi(\tau,x_{l})\right| + \left|\partial_\tau^2\phi(\tau,x_{l}) - \partial_\tau^2\phi(\tau_{l},x_{l})\right|\:.
\end{equation}
By Sobolev embedding theorem, for $k>(D-1)/2$ we have $\partial_\tau^2\phi$ to be a continuous function in $x$,  and thus, the first term on the right hand side vanishes as $l\rightarrow \infty$. The second term can be estimated by a constant times $\|\partial_\tau^2\phi(\tau,\cdot)-\partial_\tau^2\phi(\tau_{l},\cdot)\|_{H^{k-1}(\lR^{D-1})}$. Since $\partial_\tau^2\phi \in C\left([\tau_0,\tau_0 + T],H^{k-1}(\lR^{D-1})\right)$, then the second term  also goes to zero as $l\rightarrow \infty$. Hence, we conclude that, $\partial_\tau^2\phi \in C\left([\tau_0,\tau_0 + T]\times\lR^{D-1}\right)$, thus
\begin{equation}
\phi \in C^2\left([\tau_0,\tau_0 + T]\times\lR^{D-1}\right)\:.
\end{equation}

To show the uniqueness, consider $\phi,\phi'$ as solutions of Equation (\ref{MainEq}) with the same initial data. Similar to the proof of Lemma \ref{energyestimate}, we have the estimate
\begin{eqnarray}
\frac{d}{d\tau}\mathcal{H}_k[\phi'-\phi] & \leq & C \|F(\phi',\partial_\tau \phi') - F(\phi,\partial_\tau \phi)\|_{H^{k}(\lR^{D-1})}\;\|\partial_{\tau}\phi' - \partial_{\tau}\phi\|_{H^{k}(\lR^{D-1})} \nonumber\\
& \leq & C \mathcal{H}_k[\phi'-\phi]\:.
\end{eqnarray}
Using Gronwall lemma and the fact that $\phi,\phi'$ have the same initial data, then for all $\tau \in [\tau_0,\tau_0 + T]$, we conclude that $\phi'(\tau) = \phi(\tau)$ and the uniqueness follows.

Thus, we have proven,
\begin{theorem}
	\label{theoremexistence}
	Let $f \in H^{k+1}(\lR^{D-1})$ and $g \in H^{k}(\lR^{D-1})$ be initial data with compact support. Assume that the scalar potential is a smooth function satisfying $V(0) = 0$ and $\partial_{\phi}V(0) = 0$ and that $k>(D-1)/2$. Then, there exist $T>0$ and a unique $\phi \in C^2\left([\tau_0,\tau_0 + T]\times\lR^{D-1}\right)$ being a local solution to the Equation (\ref{iPDE}) such that
	\begin{equation}
	\phi \in C\left([\tau_0,\tau_0 + T], H^{k+1}(\lR^{D-1})\right) \cap C^1 \left([\tau_0,\tau_0 + T], H^{k}(\lR^{D-1})\right)\:.
	\end{equation}
\end{theorem}

\subsection{Smoothness Properties}

The local solution, which we have discussed above, is $\phi \in C^2\left([\tau_0,\tau_0 + T],\lR^{D-1}\right)$ throughout their interval of existence. In general, however they are actually smoother than this. In this section, we prove the smoothness properties of the solution of equation  (\ref{MainEq}).

We claim that the solution is $C^{(m-1)}\left([\tau_0,\tau_0 + T],\lR^{D-1}\right)$ for $m \geq 3$. We will prove the statement using the induction argument. For $m=3$, the statement is true. Now, assume that it is also true for $m=n$, Then, we have
\begin{equation}
\partial^r_\tau\partial_{j_1}\partial_{j_2}\cdots\partial_{j_{i-r-1}}\phi \in C\left([\tau_0,\tau_0 +T]\times\lR^{D-1}\right)\:,
\end{equation}
where $r=0,1,\cdots,n-1$. For $m = n+1$, we need to prove that the limit of sequence defined by equations (\ref{zerotheq}) and (\ref{higheq})  satisfies
\begin{equation}
\partial_{j_1}\partial_{j_2}\cdots\partial_{j_{i-p}}\partial^p_\tau\phi \in C\left([\tau_0,\tau_0 +T]\times\lR^{D-1}\right)\:, \label{ra}
\end{equation}
for $p=0,1,\cdots,n-1$.

Since $n+1 > 3$, using Lemma \ref{lemmaboundEl}, we have
\begin{equation}
E_{n+1}[\phi_{l} - \phi_{l-1}] \leq C \:,
\end{equation}
for some constant $C$. In other words, we have the estimate
\begin{equation}
\|\partial_j(\phi_{l} - \phi_{l-1})(\tau,\cdot)\|_{H^{n+1}(\lR^{D-1})} + \|\partial_\tau(\phi_{l} - \phi_{l-1})(\tau,\cdot)\|_{H^{n+1}(\lR^{D-1})} \leq C\:. \label{rb}
\end{equation}
From the first term on the right hand side, we obtain
\begin{equation}
\|\partial_{j_1}\cdots\partial_{j_{n-1}}\partial_j(\phi_{l} - \phi_{l-1})(\tau,\cdot)\|_{H^{2}(\lR^{D-1})} \leq C \|\partial_j(\phi_{l} - \phi_{l-1})(\tau,\cdot)\|_{H^{n+1}(\lR^{D-1})} \leq C\:.
\end{equation}
Using Sobolev embedding theorem, we conclude
\begin{equation}
\partial_{j_1}\cdots\partial_{j_{n}}\phi \in C\left([\tau_0,\tau_0 + T]\times\lR^{D-1}\right)\:.
\end{equation}
For the second term, using similar methods, we have
\begin{equation}
\partial_{j_1}\cdots\partial_{j_{n-1}}\partial_\tau\phi \in C\left([\tau_0,\tau_0 + T]\times\lR^{D-1}\right)\:.
\end{equation}
Hence, we have shown that Equation (\ref{ra}) is satisfied for $p=0$ and $p=1$. 

From Equation (\ref{higheq}), we have
\begin{eqnarray}
\partial^p_\tau \partial_{j_1}\cdots\partial_{j_{i-p}}(\phi_{l} - \phi_{l-1}) = \partial^{p-2}_\tau \partial_{j_1}\cdots\partial_{j_{i-p}}\hat{F}_l  + \partial^{p-2}_\tau \partial_{j_1}\cdots\partial_{j_{i-p}}\partial_k\partial^k(\phi_{l} - \phi_{l-1})\label{rf}\:,
\end{eqnarray}
with 
\begin{equation}
\hat{F}_{l} = F(\phi_l,\partial_{\tau}\phi_{l}) - F(\phi_{l-1},\partial_{\tau}\phi_{l-1})\:.
\end{equation}
The first term of Equation (\ref{rf}) can be written as
\begin{eqnarray}
\nonumber\partial^{p-2}_\tau \partial_{j_1}\cdots\partial_{j_{i-p}}\hat{F}_l &=& -\frac{12\xi}{\tau^2} \partial^{p-2}_\tau \partial_{j_1}\cdots\partial_{j_{i-p}}(\phi_{l} - \phi_{l-1}) - \frac{4}{\tau} \partial^{p-1}_\tau \partial_{j_1}\cdots\partial_{j_{i-p}}(\phi_{l} - \phi_{l-1})\\&&
-~ \tau^4 \partial^{p-2}_\tau \partial_{j_1}\cdots\partial_{j_{i-p}}\left[\partial_{\phi}V(\phi_l) - \partial_{\phi}V(\phi_{l-1})\right]\:.
\end{eqnarray}
From the induction hypothesis, the first and second terms are $C\left([\tau_0,\tau_0 + T]\times\lR^{D-1}\right)$. To estimate the last term, define
\begin{equation}
G(\phi_l,\phi_{l-1}) = \int_{0}^{1}\partial_\sigma \partial_{\phi}V[\sigma\phi_l + (1-\sigma)\phi_{l-1}]\:d\sigma\:,
\end{equation}
and we have the estimate,
\begin{eqnarray}
\nonumber\partial^{p-2}_\tau \partial_{j_1}\cdots\partial_{j_{i-p}}\left[G(\phi_l,\phi_{l-1})\hat{\phi}_{l-1}\right] &=& \sum_{c+d\leq i-p}^{}~~\sum_{s+t\leq p-2}^{}(\partial^s_\tau\partial_{j_1}\cdots\partial_{j_c}G)\\
&&~~~~~~~~~~~~~~~~~~~~~\nonumber(\partial^t_\tau\partial_{j_{c+1}}\cdots\partial_{j_{c+d}}(\phi_{l} - \phi_{l-1}))\:,\\
\end{eqnarray}
for $c\geq 1$ and $d\geq 0$. Since the scalar potential is a smooth function and $E_c[\phi_l]$ is bounded for $c\leq n-p$ by Lemma \ref{iterationboundenergy}, therefore we can bound $(\partial^s_\tau\partial_{j_1}\cdots\partial_{j_c}G)$ by a constant. Because the last factor is $C\left([\tau_0,\tau_0 + T]\times\lR^{D-1}\right)$ by induction hypothesis, then $\partial^{p-2}_\tau \partial_{j_1}\cdots\partial_{j_{i-p}}\hat{F}_l$ is also $C\left([\tau_0,\tau_0 + T]\times\lR^{D-1}\right)$, and we conclude 
\begin{eqnarray}
\partial^p_\tau\partial_{j_1}\partial_{j_2}\cdots\partial_{j_{i-p}}\phi \in C\left([\tau_0,\tau_0 + T]\times\lR^{D-1}\right)\:,
\end{eqnarray}
for $p=0,1,\cdots,n-1$.

Hence, we have proven

\begin{theorem}
	\label{theoremsmoothness}
	Let $f \in H^{k+1}(\lR^{D-1})$ and $g \in H^{k}(\lR^{D-1})$ be the initial data with compact support. Assume that the scalar potential is a smooth function satisfying $V(0) = 0$ and $\partial_{\phi}V(0) = 0$ and that $k>(D-1)/2$. Then, there exist $T>0$ and a unique $\phi \in C^{(k-1)}\left([\tau_0,\tau_0 + T]\times\lR^{D-1}\right)$, which is a local solution to the Equation (\ref{iPDE}), such that
	\begin{equation}
	\phi \in C\left([\tau_0,\tau_0 + T], H^{k+1}(\lR^{D-1})\right) \cap C^1 \left([\tau_0,\tau_0 + T], H^{k}(\lR^{D-1})\right)\:.
	\end{equation}
\end{theorem}

\section{Global Existence}
\label{sec:GlobalEx}

\subsection{General Setup} 
	
In Section \ref{sec:LocalEx}, we proved the existence a unique local solution of Equation (\ref{iPDE}), i.e.\ 
\begin{equation}
\begin{cases}
\partial_{\tau}^2 \phi + (D-2) H ~ \partial_\tau\phi - \Delta\phi = - \xi (D-1)\left( 2 \dot{H}  + (D-2) H^2\right)\phi + a^2 \partial_{\phi}V(\phi)  \\
\phi(\tau_{0},\cdot)  =  f \in H^{k+1}(\lR^{D-1}) \\
\partial_{\tau}\phi(\tau_{0},\cdot)  =  g \in H^{k}(\lR^{D-1}) \:.
\end{cases}
\end{equation}
Here, we show that it is possible to have a set of global solutions of (\ref{iPDE}) for $T \to +\infty$. To proceed, let us  define a new field, $\psi = a^{(D-2)/2} \phi$ where $a \equiv a(\tau, \tau_0)$ such that Equation (\ref{iPDE}) can be written down as 
\begin{equation}
\begin{cases}
\partial_{\tau}^2 \psi - \Delta\psi = h(\tau) \psi + a^D \partial_{\psi}V(\psi) \\
\psi(0,\cdot)  = a(\tau_0)^{(D-2)/2} f  \\
\partial_{\tau}\psi(0,\cdot)  = \frac{1}{2}(D-2)  a(\tau_0)^{(D-4)/2} \dot{a}(\tau_0) f +  a(\tau_0)^{(D-2)/2} g  \:,
\end{cases}
\label{transformEq}
\end{equation}
where
\begin{equation}
h(\tau) = \left[  \frac{1}{2}(D-2) -2 \xi (D-1)\right]\left[\dot{H} + \frac{D-2}{2}H^2\right] ~  . \label{Riccatieq}
\end{equation}
Note that Equation (\ref{Riccatieq}) is Riccati's form of the Hubble parameter $H(\tau)$. In particular, $H(\tau)$ could be thought of as a solution of Riccati's equation, see for example,  \cite{Polyaninbook:2003}.
%
In the rest of the paper we simply take several assumptions as follows. 
\begin{assumption}\label{assump1}
$h(\tau) \leq 0$ and $\partial_{\tau}h(\tau) \geq 0$ for all $\tau \in \lR^{+}$. 
\end{assumption}
\noindent This assumption follows that $h(\tau)$ tends to vanish as $\tau \to + \infty$. For example, the function $h(\tau)$  may have the form of either  $-m \mathrm{e}^{-n\tau}$ or $-m \tau^{-n}$ with $m, n \in \lR^{+}$. The latter function for $n=2$ could be related to a cosmological model where the scale factor $a(\tau)$ has a polynomial form. This occurs, for example, in  the standard cosmological models   discussed in Section \ref{sec:SpatiallyFlatSpacetimes}.
%
\begin{assumption}
The scalar potential has the form 
\begin{equation}
V(\phi) =  -\frac{\epsilon}{p+1} \:\phi^{p+1}\:,
\end{equation}
with $\epsilon$ is a small positive parameter and $p \in \lR^{+}$.
\end{assumption}
\begin{assumption}\label{assump3}
 $\int_{0}^{\infty}  a^{[D+2 - (D-2)p]/2}   d\tau = A(\tau_0) < +\infty $ for all $\tau_0 \in \lR^+$. 
\end{assumption}

Thus, we can write the differential equation as
\begin{equation}
\partial_{\tau}^2 \psi - \Delta\psi = h(\tau)\psi + \epsilon P(\tau,\psi) \:, \label{globalnoneq}
\end{equation}
where $h(\tau)$ is given by (\ref{Riccatieq}) and  $P(\tau,\psi) =  - a^{[D+2 - (D-2)p]/2} ~ \psi^{p} $.

Now, suppose we have  the nonhomogenous linear equation
\begin{equation}
\begin{cases}
\partial_{\tau}^2 \eta - \Delta\eta = h(\tau)\eta + P(\tau,x) \\
\eta(0,\cdot)  =  a(\tau_0)^{(D-2)/2} f  \\
\partial_{\tau}\eta(0,\cdot)  = \frac{1}{2}(D-2)  a(\tau_0)^{(D-4)/2} \dot{a}(\tau_0) f + a(\tau_0)^{(D-2)/2} g  \:,
\end{cases} \label{nonhomlineareq}
\end{equation}
where $h(\tau)$ is given by (\ref{Riccatieq}) and $\tau_0 >0$ is arbitary real number. First, we prove the following lemma,
\begin{lemma}
	\label{lemmanonhomlin}
	Let $\eta$ be a solution of linear Equation (\ref{nonhomlineareq}) with compact support. Let $P(\tau,x)$ be a $C^0\left(\lR^{+},H^{k}(\lR^{D-1})\right)$ function such that $\int_{0}^{\infty}\:\|P(s)\|_{H^{k}(\lR^{D-1})}\:ds <\infty$. If Assumption \ref{assump1} holds, then we have the following inequality,
	\begin{equation}
	\sup_{\tau \in \lR^{+}}\: |\eta(\tau)|_{k} \leq C\left(\|g\|_{H^{k}(\lR^{D-1})} + \| f\|_{H^{k+1}(\lR^{D-1})} + \int_{0}^{\infty}\|P(s)\|_{H^{k}(\lR^{D-1})}\:ds\right) \leq M\:, \label{nonhomlineqineq}
	\end{equation}
	where $M$ depends on the initial data, $k$, $h$ and $P$.
\end{lemma} 

\begin{proof}
	Taking spatial Fourier transform of Equation (\ref{nonhomlineareq}), we have
	\begin{equation}
	\partial_{\tau}^2\tilde{\eta} + |\lambda|^2 \tilde{\eta} = h(\tau) \tilde{\eta} + \tilde{P}\:.
	\end{equation}
	Multipying by $\partial_{\tau}\tilde{\eta}$, integrating over $\tau$ and using partial integration, we obtain
	\begin{eqnarray}
	|\partial_{\tau}\tilde{\eta}(\tau)|^2 + |\lambda|^2 |\tilde{\eta}(\tau)|^2 & = & |\tilde{g}|^2 + |\lambda|^2 |\tilde{f}|^2 + h(\tau)|\tilde{\eta}(\tau)|^2 - h(0)|\tilde{f}|^2 \nonumber\\
	& &\qquad - \int_{0}^{\tau} \partial_{\tau}h(s) |\tilde{\eta}(s)|^2\;ds + 2 \int_{0}^{\tau} \tilde{P}(s) \partial_{\tau}\tilde{\eta}(s)\;ds \:.
	\end{eqnarray}
	
	Since we have $h(\tau) \leq 0$ and $\partial_{\tau}h(\tau) \geq 0$ for all $\tau \in \lR^{+}$, then  we have
	\begin{equation}
	|\partial_{\tau}\tilde{\eta}(\tau)|^2 + |\lambda|^2 |\tilde{\eta}(\tau)|^2 \leq |\tilde{g}|^2 + |\lambda|^2 |\tilde{f}|^2 - h(0)|\tilde{f}|^2 + 2 \int_{0}^{t} \tilde{P}(s) \partial_{\tau}\tilde{\eta}(s)\;ds \:.
	\end{equation}
	Multiplying by $(1+\lambda^2)^k$ and integrating over $\lambda$-space, we obtain
	\begin{equation}
	\mathcal{H}_k[\eta](\tau) \leq C\left(\|g\|^{2}_{H^{k}(\lR^{D-1})} + \| f\|^{2}_{H^{k+1}(\lR^{D-1})} + 2\int_{0}^{\tau}\int_{\lR^{D-1}}\:(1+\lambda^2)^k\tilde{P}(s)\partial_{\tau}\tilde{\eta}(s)\;d\lambda\;ds\right)\:,
	\end{equation}
	where $C$ depends only on $k$ and $h$. Using H\"older inequality,
	\begin{equation}
	\int_{\lR^{D-1}}\:(1+\lambda^2)^k\tilde{P}(s)|\partial_{\tau}\tilde{\eta}(s)|^2\;d\lambda \leq \left[\int_{\lR^{D-1}}\:(1+\lambda^2)^k|\tilde{P}(s)|^2\;d\lambda\right]^{1/2}\left[\int_{\lR^{D-1}}\:(1+\lambda^2)^k|\partial_{\tau}\tilde{\eta}(s)|^2\;d\lambda\right]^{1/2}
	\end{equation}
	we have,
	\begin{equation}
	\mathcal{H}_k[\eta](\tau) \leq C\left(\|g\|^{2}_{H^{k}(\lR^{D-1})} + \| f\|^{2}_{H^{k+1}(\lR^{D-1})} + \sup_{s \in [0,\tau]}\|\partial_{\tau}\eta(s)\|_{H^{k}(\lR^{D-1})}\int_{0}^{\tau}\:\|P(s)\|_{H^{k}(\lR^{D-1})}\;ds\right)\:. \label{lemmanonhomlin1}
	\end{equation}
	
	Let us consider,
	\begin{eqnarray}
	\frac{\partial}{\partial \tau}\int_{\lR^{D-1}}\:|\eta|^2\:dx & = & 2 \int_{\lR^{D-1}}\:\eta\;\partial_{\tau}\eta\:dx \nonumber \\
	& \leq & 2 \|\eta\|_{L^{2}(\lR^{D-1})} \|\partial_{\tau}\eta\|_{L^{2}(\lR^{D-1})} \leq C \|\eta\|_{L^{2}(\lR^{D-1})}\mathcal{H}^{1/2}_{k}[\eta]\:.
	\end{eqnarray}
	Integrating the inequality, we obtain
	\begin{equation}
	\|\eta(\tau)\|^{2}_{L^{2}(\lR^{D-1})} \leq \|\eta(0)\|^{2}_{L^{2}(\lR^{D-1})} + C \int_{0}^{\tau}\|\eta(s)\|_{L^{2}(\lR^{D-1})}\mathcal{H}^{1/2}_{k}[\eta](s)\:ds\:. \label{lemmanonhomlin2}
	\end{equation}
	Adding up inequalities (\ref{lemmanonhomlin1}) and (\ref{lemmanonhomlin2}) and taking the supremum for $\tau \in \lR^{+}$ will yield (\ref{nonhomlineqineq}) and the proof is finished.
\end{proof}

Since $\psi = a^{(D-2)/2} \phi$, by Theorems \ref{theoremexistence} and \ref{theoremsmoothness}, there exists a unique local solution of Equation (\ref{transformEq}). Furthermore, by Lemma (\ref{lemmanonhomlin}), for $\tau \in [0,T]$ we have the following inequality,
\begin{equation}
 |\psi(\tau)|_{k} \leq C\left(\|g\|_{H^{k}(\lR^{D-1})} + \| f\|_{H^{k+1}(\lR^{D-1})} + \int_{0}^{\tau}\|P(s,\psi(s))\|_{H^{k}(\lR^{D-1})}\:ds\right)\:. \label{psilocal}
\end{equation} 
We show that the result can be extended to $T\rightarrow \infty$, hence the solution globally exists.

Similar with proving the local existence, we construct a sequence $\{\psi_{l}(\tau)\}$ such that,
\begin{equation}
\begin{cases}
\partial_{\tau}^2 \psi_0 - \Delta\psi_0 = h(\tau)\psi_0 + \epsilon P(\tau,0) \\
\partial_{\tau}^2 \psi_{l+1} - \Delta\psi_{l+1} = h(\tau)\psi_{l+1} + \epsilon P(\tau,\psi_{l})\\
\psi_{l}(0,\cdot)  = a^{(D-2)/2}  f  \\
\partial_{\tau}\psi_{l}(0,\cdot)  =  \frac{1}{2}(D-2)  a(\tau_0)^{(D-4)/2} \dot{a}(\tau_0) f + a^{(D-2)/2}  g  \:,
\end{cases}
\end{equation}
We show by induction that there exists a positive constant $\epsilon_0$ such that for $0<\epsilon<\epsilon_0$, we have
\begin{eqnarray}
\sup_{\tau \in \lR^{+}}|\psi_l(\tau)|_k & \leq & M \label{globalitebound}\\
\sup_{\tau \in \lR^{+}}|\psi_{l+1}(\tau)-\psi_l(\tau)|_k & \leq & \kappa \sup_{\tau \in \lR^{+}}|\psi_{l}(\tau)-\psi_{l-1}(\tau)|_k\:, \label{globaliteCauchy}
\end{eqnarray}
for some positive constant $\kappa \in (0,1)$. For $l=0$, the inequalities are true due to Lemma \ref{lemmanonhomlin}. Assume that these are true for $l=n$, thus similar to the proof of Lemma \ref{nonlinearbound}, we have
\begin{eqnarray}
\|P(\tau,\psi_n(\tau))\|_{H^{k}(\lR^{D-1})} & \leq & C  a^{[D+2 - (D-2)p]/2} \|\psi_n(\tau)\|_{H^{k+1}(\lR^{D-1})} \nonumber\\
& \leq & C  a^{[D+2 - (D-2)p]/2}  |\psi_n(\tau)|_k \nonumber\\
& \leq &  C M  a^{[D+2 - (D-2)p]/2} 
\end{eqnarray}
Then, from (\ref{nonhomlineqineq}) we obtain
\begin{eqnarray}
\sup_{\tau \in \lR^{+}}\: |\psi_{n+1}(\tau)|_{k} & \leq & C_0\left(\|g_{n+1}\|_{H^{k}(\lR^{D-1})} + \| f_{n+1}\|_{H^{k+1}(\lR^{D-1})} + \epsilon \int_{0}^{\infty}\|P(s,\psi_n(s))\|_{H^{k}(\lR^{D-1})}\:ds\right) \nonumber\\
& = & C_0\left(L + \epsilon C M \int_{0}^{\infty}   a^{[D+2 - (D-2)p]/2}   ds \right)\nonumber\\
& \leq & C_0\left(L +  \epsilon  C M A(\tau_0) \right)\:.
\end{eqnarray}
Defining $\epsilon_0 = (M-C_0 L) /C_0 C M A(\tau_0)$, it implies
\begin{equation}
\sup_{\tau \in \lR^{+}}\: |\psi_{n+1}(\tau)|_{k} \leq M\:,
\end{equation}
and, hence the inequality (\ref{globalitebound}) holds for all nonnegative integers.

Similar to the proof of Lemma \ref{energyestimate}, we have
\begin{equation}
\|P(\tau,\psi_{l}(\tau))-P(\tau,\psi_{l-1}(\tau))\|_{H^{k}(\lR^{D-1})} \leq Ca^{[D+2 - (D-2)p]/2} |\psi_{l}(\tau)-\psi_{l-1}(\tau)|_k\:,
\end{equation}
which follows that
\begin{eqnarray}
\sup_{\tau \in \lR^{+}}\: |\psi_{l+1}(\tau) - \psi_{l}(\tau)|_{k} & \leq & C_0 \epsilon \int_{0}^{\infty}\|P(s,\psi_l(s)) - P(s,\psi_{l-1}(s)) \|_{H^{k}(\lR^{D-1})}\:ds \nonumber\\
&\leq & C_0 C \epsilon \int_{0}^{\infty}  a^{[D+2 - (D-2)p]/2}   |\psi_{l}(\tau)-\psi_{l-1}(\tau)|_k ~ ds\nonumber\\
& \leq & C_0 C \epsilon A(\tau_0) \sup_{\tau \in \lR^{+}}\: |\psi_{l}(\tau) - \psi_{l-1}(\tau)|_{k}\:.
\end{eqnarray}
This inequality proves that $\{\psi_l\}$ converges to  $\psi \in C\left(\lR^{+},{H^{k+1}(\lR^{D-1})}\right)$. Furthermore, similar to the proof of local existence, we conclude that there exists a unique global solution of Equation (\ref{globalnoneq}) such that $\psi \in C^2\left(\lR^{+},{H^{k+1}(\lR^{D-1})}\right)$ and $\sup_{\tau \in \lR^{+}} |\psi(\tau)| \leq M$.

Hence, we have proven,
\begin{theorem}
	\label{theoremglobalexistence}
	Let $f \in H^{k+1}(\lR^{D-1})$ and $g \in H^{k}(\lR^{D-1})$ be the initial data with compact support for $k > (D-1)/2$. Suppose that Assumptions \ref{assump1}-\ref{assump3}. hold. For any positive constant $M$ that depends on the initial data, $k$ and $\tau_0$, there exists a positive number $\epsilon_0$ that also depends on the initial data, $k$ and $\tau_0$, such that for any $0<\epsilon<\epsilon_0$, Equation (\ref{iPDE}) admits unique classical global solutions
	\begin{equation}
	\phi \in C\left([\tau_0,\infty], H^{k+1}(\lR^{D-1})\right) \cap C^1 \left([\tau_0,\infty], H^{k}(\lR^{D-1})\right)\:,
	\end{equation}
	satisfying the following decay estimate
	\begin{equation}
	\left\|\phi(\tau)\right\|_{H^{k+1}(\lR^{D-1})} + \frac{1}{2} (D-2) \left\| \frac{\dot{a}}{a} \phi(\tau) + \partial_{\tau}\phi(\tau)\right\|_{H^{k}(\lR^{D-1})} \leq M  a^{-(D-2)/2} ~ ,
	\end{equation}
	where $a \equiv a(\tau, \tau_0)$.
\end{theorem}

\section{Some Models}
\label{sec:Models}

In this section, we consider some specific models related to the scale factor $a(\tau)$. These models may have a global regular solution in the sense of our setup in the preceding section.

\subsection{Power Form}

First, let us take the scale factor $a(\tau)$ to be of the form
\begin{equation}
a(\tau) = (\tau + \tau_{0})^{\alpha}\:,
\end{equation}
where $\alpha$ is a real constant. In particular, for a single component universe the constant $\alpha$ is given by  
\begin{equation}
\alpha = \frac{2}{(D-1)(w+1) - 2}\:,
\end{equation}
where $w \neq -(D-3)/(D-1)$ is related to the equation of state discussed in Section \ref{sec:SpatiallyFlatSpacetimes}. Thus we have,
\begin{equation}
h(\tau) = \frac{\alpha}{(\tau + \tau_{0})^2}\left[  \frac{1}{2}(D-2) -2 \xi (D-1)\right]\left[-1 + \frac{\alpha(D-2)}{2}\right]
\end{equation}
In Table \ref{tab2} we list four cases single component of higher dimensional cosmological models where Assumptions \ref{assump1} and \ref{assump3} are fulfilled.

\begin{table}[!htpb]
	\centering
		\begin{tabular}{| l | l | l | l | l |}
			\hline & & & & \\
			Case I & $\alpha < 0$ & $w < -\frac{D-3}{D-1}$ & $\xi > \frac{D-2}{4(D-1)}$ & $p < \frac{1}{D-2}\left(D+2 - \frac{2}{|\alpha|}\right)$ \\
			\hline  & & & &\\
			Case II & $0< \alpha < \frac{2}{D-2}$ & $w > \frac{1}{D-1}$ & $\xi < \frac{D-2}{4(D-1)}$ & $p > \frac{1}{D-2}\left(D+2 + \frac{2}{\alpha}\right)$ \\
			\hline  & & & &\\
			Case III & $\alpha > \frac{2}{D-2} $ & $-\frac{D-3}{D-1} < w < \frac{1}{D-1}$ & $\xi > \frac{D-2}{4(D-1)}$ & $p > \frac{1}{D-2}\left(D+2 + \frac{2}{\alpha}\right)$ \\
			\hline & & & &\\
			Case IV & $\alpha = \frac{2}{D-2} $ & $ w = \frac{1}{D-1}$ & $\xi \in \lR$ & $p > \frac{2D}{D-2}$ \\
			\hline			
		\end{tabular}
			\caption{Four cases in a single component higher dimensional universe.}
			\label{tab2}
\end{table}
Moreover, it is of interest to consider, for examples, the standard cosmology in four dimensions as listed in Table \ref{tab3}.

\begin{table}[!htpb] 
	\centering
	\begin{tabular}{| l | l | l | l | l |}
		\hline & & & & \\
		Matter Dominated & $w = 0$ & $\alpha = 2$ & $\xi > \frac{1}{6}$ & $p > \frac{7}{2}$ \\
		\hline  & & & &\\
		$\Lambda$-Dominated & $w = -1$ & $\alpha = -1$ & $\xi > \frac{1}{6}$ & $p < 2$ \\
		\hline  & & & &\\
		Radiation Dominated & $w = \frac{1}{3}$ & $\alpha = 1$ & $\xi \in \lR$ & $p > 4$ \\
		\hline
	\end{tabular}\
		\caption{Three cases in a single component four dimensional universe.}
		\label{tab3}
\end{table}

\subsection{Exponential Form}
Finally, we take a case where the scale factor $a(\tau)$ has the form
\begin{equation}
a(\tau) = e^{\alpha (\tau + \tau_{0})}\:,
\end{equation}
with  $\alpha \in \lR^{+}$. This exponential form is related to the single component universe for $w =  -(D-3)/(D-1)$. Then we have
\begin{equation}
h(\tau) = \frac{\alpha^2(D-2)}{2}\left[  \frac{1}{2}(D-2) -2 \xi (D-1)\right].
\end{equation}
In order to satisfy Assumptions \ref{assump1} and \ref{assump3}, we should have $w = -\frac{D-3}{D-1}$, $\xi > \frac{D-2}{4(D-1)}$, and $p > \frac{D+2}{D-2}$.

\section{Acknowledgments}

The work of this paper is supported by Riset Unggulan PT Kemenristekdikti 2017-2018.


\end{document}